\documentclass[12pt]{article}

\usepackage{theorem}
\usepackage{subfigure}
\usepackage{epsfig}
\usepackage{amssymb}
\usepackage{fullpage}
\usepackage{amsmath}
\usepackage{ifthen}
\usepackage{verbatim}
\usepackage{hyperref}
\usepackage{xspace}
\usepackage{bm}
\usepackage{setspace}
\usepackage{tablefootnote}

\newtheorem{theorem}	 			{Theorem}[section]

\newtheorem{corollary}		[theorem]	{Corollary}
\newtheorem{prop}		[theorem]	{Proposition}

{\theorembodyfont{\rmfamily} \newtheorem{definition}
[theorem]	{Definition}}
{\theorembodyfont{\rmfamily} \newtheorem{remark}		[theorem]
{Remark}}
{\theorembodyfont{\rmfamily} \newtheorem{example}		[theorem]
{Example}}
{\theorembodyfont{\rmfamily} }
{\theorembodyfont{\rmfamily} }
{\theorembodyfont{\rmfamily} }
{\theorembodyfont{\rmfamily} }
{\theorembodyfont{\rmfamily} }
{\theorembodyfont{\rmfamily} }
{\theorembodyfont{\rmfamily} }
\theoremstyle{break}
{\theorembodyfont{\rmfamily} }

\newenvironment{proof}{\noindent {\em {Proof:}}}{$\blacksquare$\vskip
\belowdisplayskip}

\newcommand{\prob}[2][]{\text{\bf Pr}\ifthenelse{\not\equal{}{#1}}{_{#1}}{}\!\left[#2\right]}
\newcommand{\expect}[2][]{\text{\bf E}\ifthenelse{\not\equal{}{#1}}{_{#1}}{}\!\left[#2\right]}

\newcommand{\tfm}{(\allocs,\prices,\burns)}

\newcommand{\bid}{b}
\newcommand{\bids}{{\mathbf \bid}}

\newcommand{\bidt}[1][t]{{\bid_{#1}}}

\newcommand{\val}{v}
\newcommand{\vals}{{\mathbf \val}}

\newcommand{\valt}[1][t]{{\val_{#1}}}

\newcommand{\alloc}{x}
\newcommand{\allocs}{{\mathbf \alloc}}

\newcommand{\alloct}[1][t]{\alloc_{#1}}

\newcommand{\price}{p}
\newcommand{\prices}{{\mathbf \price}}
\newcommand{\pricet}[1][t]{\price_{#1}}

\newcommand{\burn}{q}
\newcommand{\burns}{{\mathbf \burn}}
\newcommand{\burnt}[1][t]{\burn_{#1}}

\newcommand{\nine}{*}

\newcommand{\history}{\mathbf{H}}

\title{Transaction Fee Mechanism Design\thanks{Date of initial
journal submission: August 19, 2021. A four-page summary of this paper
    appeared in ACM SIGecom Exchanges, June 2021.
A one-page abstract of this paper
    appeared in the Proceedings of the 
22nd Annual ACM Conference on Economics and Computation, July 2021.}}
\author{Tim Roughgarden\thanks{Department of Computer Science,
    Columbia University. 
The author's research at Columbia University is supported in part by
NSF awards CCF-2006737 and CNS-2212745. The author's work on a
related report for the public~\cite{eip1559full} was supported by the
Decentralization Foundation.
Since January 1, 2022, the author has also served as Head of
Research at a16z Crypto, a
venture capital firm with investments in blockchain protocols.
Email: \texttt{tim.roughgarden@gmail.com}.}}

\begin{document}

\maketitle

\begin{abstract}
Demand for blockchains such as Bitcoin and Ethereum is far larger than
supply, necessitating a mechanism that selects a subset of 
transactions to include ``on-chain'' from the pool of all pending
transactions.
This paper investigates the problem of designing a blockchain
transaction fee mechanism through the lens of mechanism design.
We introduce two new forms of incentive-compatibility that
capture some of the idiosyncrasies of the blockchain setting, one
(MMIC) that protects against  deviations by profit-maximizing miners
and one (OCA-proofness) that protects against off-chain collusion
between miners and users.

This study is immediately applicable to major change (made on August
5, 2021) to Ethereum's transaction fee mechanism, based on a proposal
called ``EIP-1559.''  Originally, Ethereum's transaction fee
mechanism was a first-price (pay-as-bid) auction.  EIP-1559 suggested
making several tightly coupled changes, including the introduction of
variable-size blocks, a history-dependent reserve price, and the
burning of a significant portion of the transaction fees.  We prove
that this new mechanism earns an impressive report card: it satisfies
the MMIC and OCA-proofness conditions, and is also dominant-strategy
incentive compatible (DSIC) except when there is a sudden demand
spike.  We also introduce an alternative design, the ``tipless
mechanism,'' which offers an incomparable slate of
incentive-compatibility guarantees---it is MMIC and DSIC, and
OCA-proof unless in the midst of a demand spike.
\end{abstract}

\section{Introduction}\label{s:intro}

Real estate on a major blockchain is a scarce resource.  For example,
Bitcoin~\cite{bitcoin} and Ethereum~\cite{ethereum}, the two biggest
blockchains by market
cap\footnote{\url{https://www.coingecko.com/}},
process roughly~5
and~15 transactions per second on average, respectively.
Demand for these blockchains is far larger, necessitating a mechanism
that selects a subset of transactions to include ``on-chain'' from the
pool of all submitted transactions.

Historically, most blockchain protocols have employed a pay-as-bid
transaction fee mechanism.  Every transaction is submitted with a bid
(in the blockchain's native currency), the miner of a block decides
which transactions should be included in it, and upon publication of
that block, the bid of each included transaction is transferred from
its creator to the miner.\footnote{``Miner'' is the common term for
  block producers in a proof-of-work blockchain protocol (in which
  block validity rests on the inclusion of a partial pre-image of a
  cryptographic hash function). At the time the work described in this
  paper was done (in 2020--2021), Ethereum was a proof-of-work
  blockchain protocol. On September 15, 2022, in an event now known as
  ``the Merge,'' the protocol upgraded to an alternative method of
  sybil-resistance called ``proof-of-stake.'' In a proof-of-stake
  blockchain protocol, block producers are usually called
  ``validators'' rather than ``miners.'' Because the model in this
  paper concerns the production of a single block and takes the
  identity of the block producer as fixed (agnostic to how it was
  selected), the Ethereum protocol's switch from proof-of-work to
  proof-of-stake does not change any of the incentive-compatibility
  results proved here.\label{foot:pos}}  We follow blockchain convention and refer to
this mechanism as a {\em first-price auction (FPA)}.

FPAs are natural enough and served for many years as the dominant
paradigm in blockchain protocols, but are they really the best we can
do?
Could a different transaction fee mechanism offer more compelling
incentive-compatibility properties?

The primary goal of this paper is to investigate 
these questions through the lens of mechanism design,
while taking into account the many idiosyncrasies of
the blockchain setting relative to more traditional applications of
the field.  For example:
\begin{enumerate}

\item [(1)] The miner of a block has dictatorial control over its contents,
  and in particular may deviate from the allocation rule intended by
  the protocol designer.

\item [(2)] The miner of a block can costlessly include fake transactions
  that are indistinguishable from real transactions.

\item [(3)] Payments should be computable from ``on-chain'' data,
which typically discloses no information about losing bids.

\item [(4)] Miners and users can easily collude off-chain to manipulate a
  transaction fee mechanism.

\end{enumerate}

The sequential and repeated nature of the blockchain setting also
offers some potential advantages to the mechanism designer (which are
not exploited by FPAs).  For example:
\begin{itemize}

\item [(5)] The choice of mechanism (such as a reserve price) for a given
  block could be informed by the (publicly visible) outcomes for
  previous blocks.

\item [(6)] Revenue from a block need not be transferred directly to the
  block's miner and could instead be redirected, for example to 
a foundation or to the
  miners of future blocks.

\end{itemize}
The key question is then: is there a transaction fee mechanism,
possibly taking advantage of points~(5) and~(6), that meets the
constraints imposed by~(1)--(4) while also decreasing the
strategic complexity relative to FPAs?

\subsection*{EIP-1559}

In addition to its basic scientific interest, the analysis of
transaction fee mechanisms is immediately applicable to scrutinizing
what is arguably the biggest change made to-date to the Ethereum
blockchain.  As background, a blockchain can change its own
specification through a ``hard fork,'' meaning a coordinated switch by
the nodes in its network to a new and backwards-incompatible version
of the protocol software.  At the time of this writing, the Ethereum
blockchain has had roughly fifteen hard forks since its genesis in
May~2015.  For example, the ``London fork'' began with block number
12865000, which was mined on August 5, 2021.  Each hard fork
implements a collection of Ethereum improvement protocols or ``EIPs''
that have been vetted by the community and approved for inclusion.

One of the EIPs implemented in the London fork is EIP-1559, a proposal
by Vitalik Buterin (Ethereum's
founder)~\cite{vb1559,1559spec} that suggested several
tightly coupled changes to Ethereum's transaction fee mechanism (which was
previously an FPA), including the introduction of variable-size blocks,
a history-dependent reserve price, and the burning of a significant
portion of the transaction fees.\footnote{The results of the new
  transaction fee mechanism can be tracked at
  \url{https://ultrasound.money/}.}
While Buterin did not provide a formal economic analysis of his proposed
design, in~\cite{vb_ec18} he outlined the motivation
behind the proposal:
\begin{quote}
Our goal is to
  discourage the development of complex miner strategies and complex
  transaction sender strategies in general, including both complex
  client-side calculations and economic modeling as well as various
  forms of collusion.
\end{quote}

Community discussion of EIP-1559 began in earnest in
early~2018\footnote{\url{https://ethereum-magicians.org/t/eip-1559-fee-market-change-for-eth-1-0-chain/2783}}
and
over time the proposal garnered a number of advocates and critics.
The polarization around this EIP was evident in the community
survey conducted by Tim Beiko.\footnote{\url{https://medium.com/ethereum-cat-herders/eip-1559-community-outreach-report-aa18be0666b5}}
``Difficulties analyzing the EIP'' ranked among the chief risks
pointed out by the survey respondents, citing ``the lack of a formal
specification or proof for the mechanism that people can independently
evaluate and critique.''
As part of the course of our work, we formalize the transaction fee
mechanism proposed in EIP-1559 and rigorously interrogate to what
extent it meets its stated design goals.  Anecdotal evidence suggests that a
preliminary report based on this work~\cite{eip1559full} contributed to the
understanding of and broader support for EIP-1559.\footnote{See, for
  example, \url{https://thedefiant.io/eip-1559-user-dev-changes/},
  \url{https://medium.com/centaur/centaur-library-how-eip-1559-will-lower-high-ethereum-gas-fees-a6a9b4f9583c},
  and \url{https://cryptonews.com/news/ethereum-s-hope-no-1559-what-it-does-and-what-it-doesn-t-do-11253.htm}.} 

We stress that while EIP-1559 provides important and timely motivation
for the present work, the discussion and results in this paper apply
equally well to many other blockchains, including Bitcoin.  The
Bitcoin community is famously hostile to major changes to the
Bitcoin protocol, however, so its transaction fee mechanism 
is likely to remain an FPA for the foreseeable future.
Some smaller blockchains have deployed variations of the
mechanism proposed in EIP-1559,
including Filecoin\footnote{See
  \url{https://filfox.info/en/stats/gas}.} and NEAR\footnote{See \url{https://near.org/papers/the-official-near-white-paper/}.}.

\subsection*{Paper Structure}

Section~\ref{s:prelim} of this paper describes our basic model, a
mechanism design setting in which the mechanism participants (creators
of transactions) compete for transaction inclusion in a block with
limited capacity.  The formalism follows in the mechanism design
tradition of allocation rules and payment rules, but with two twists.
First, to take into account point~(6) above, the payment rule
(describing transfers to a block's miner) is supplemented with a
burning rule (Definition~\ref{d:burning}) which indicates how much of
the block's revenue is burned (or otherwise redirected away from the
block's miner).  Second, in light of point~(3), payment
and burning rules are restricted to depend only on the bids of the
winning transactions.

Section~\ref{s:ic} spells out three forms of incentive-compatibility,
each guaranteeing robustness to a particular type of deviation from
the intended behavior.  Robustness to deviations from straightforward
bidding by transaction creators is captured by the familiar notion of
dominant-strategy incentive-compatibility (DSIC).  The other two
guarantees are specifically motivated by the blockchain setting and
are new to this paper.  First, we define incentive-compatibility for
myopic miners (Definition~\ref{def:mmic}), or MMIC, a condition stating that a
transaction fee mechanism should be robust to deviations by
profit-maximizing miners from the intended allocation rule (see
point~(1)) and also the injection of fake transactions (point~(2)).
Second, we define OCA-proofness (Definition~\ref{def:oca}), which
states that the mechanism should be robust to cartels of transaction
creators and miners colluding off-chain; more precisely, no off-chain
arrangement among members of such a cartel should be capable of Pareto
improving over a canonical on-chain outcome.  The rest of the paper
investigates to what extent different transaction fee mechanisms enjoy
these three strategic robustness properties.

Section~\ref{s:1559} provides a formal description of the transaction
fee mechanism proposed in EIP-1559, which we call the 1559 mechanism.
This mechanism makes use of a reserve price that is adjusted over time
in response to excess demand, and variable-size blocks to provide an
on-chain signal of excess demand.  All revenue generated by the
reserve price is burned.  Finally, there is effectively an FPA
sprinkled on top: Transaction creators also have the option of
supplementing the reserve price by an additional ``tip'' that is
transferred to the block's miner.
Small tips should be sufficient to
incentivize a miner to include a transaction during a period of stable
demand, when there is room in the current block for all the
outstanding transactions that are willing to pay the reserve price.  Large
tips can be used to encourage special treatment of a transaction, such
as immediate inclusion in a block in the midst of a sudden demand spike.
Section~\ref{s:1559} also introduces the tipless mechanism, a variant
of the 1559 mechanism in which the tips are hard-coded rather than
user-specified.  Relative to the 1559 mechanism, the tipless mechanism
provides more robustness to non-straightforward bidding while
sacrificing some resistance to off-chain collusion.

Section~\ref{s:results} contains the main results of this paper, and
considers in turn the MMIC, DSIC, and OCA-proofness properties.
Section~\ref{ss:mmic_results} provides a sufficient condition for establishing
the MMIC property (Theorem~\ref{t:mmic}): the payment of a transaction
should be independent of the other included transactions and their
bids, and the allocation rule should maximize miner profit.  This
condition implies that FPAs, the 1559 mechanism, and the tipless
mechanism are MMIC (Corollary~\ref{cor:mmic}).  Second-price-type
auctions are notable examples of non-MMIC mechanisms
(Example~\ref{ex:spa_mmic}).

Section~\ref{ss:dsic_results} studies the extent to which different
transaction fee mechanisms meet the DSIC condition.
Theorem~\ref{t:tipless_dsic} shows that the tipless mechanism acts as a
posted-price mechanism (with price equal to the reserve price plus the
hard-coded tip) and, as such, is DSIC.
The 1559 mechanism reverts to an FPA in the special case of a zero
reserve price, but Theorem~\ref{t:1559dsic} identifies two conditions that
together are sufficient for the optimality of straightforward bidding:
the base fee should not be excessively low for the current demand curve
(Definition~\ref{def:low}), and transaction creators should use
individually rational bidding strategies.  Under these conditions, the
1559 mechanism acts as a posted-price mechanism, with price equal to
the reserve price plus the miner's opportunity cost for transaction
inclusion.

Section~\ref{ss:ocaproof_results} considers OCA-proofness.
Proposition~\ref{prop:ocaproof} offers a characterization: A
transaction fee mechanism is OCA-proof if and only if there is always
an individually rational bidding strategy that leads to an outcome
that maximizes the joint utility of the transaction creators and the
block's miner.  This proposition implies that FPAs and the 1559
mechanism are both OCA-proof (Corollaries~\ref{cor:fpa_ocaproof}
and~\ref{cor:1559ocaproof}, respectively).
A key driver of the latter result is that the reserve price
in the 1559 mechanism
is determined solely by past history (leveraging the opportunity in
point~(5)) and not by a block's miner or its contents.
Corollary~\ref{cor:tipless_ocaproof} and
Remark~\ref{rem:tipless_ocaproof} show that the tipless mechanism
loses OCA-proofness exactly in the regime in which the 1559 mechanism
loses the DSIC property (that is, with an excessively low base fee).
Finally, OCA-proofness considerations show that two of the major
innovations in the 1559 mechanism, relative to an FPA---a reserve
price and burned fees---are inextricably linked, as making either
change without the other leads to a non-OCA-proof transaction fee
mechanism.
Intuitively, burning fees in an FPA catalyzes a miner and users to
collude off-chain to avoid the burn
(Corollary~\ref{cor:fpaburn_ocaproof})
while passing revenue from a reserve price to a block's miner
opens the door for low-value transactions to pay the reserve price
on-chain while receiving a partial refund from the miner off-chain
(Corollary~\ref{cor:noburn1559_ocaproof}).

Section~\ref{s:disc} concludes with a discussion of the two designs
that fare best in our study (the 1559 and tipless mechanisms), a
``pay-it-forward'' alternative to money-burning, and the possibility
of long-term collusion by cartels of miners.

\subsection*{Related Work}

\paragraph{First-price auctions.}
Transaction fee mechanisms have been an integral part of blockchain
protocol design since Nakamoto's original white paper introducing the
Bitcoin protocol~\cite{bitcoin}.  Since its genesis, Bitcoin has used
an FPA as its transaction fee mechanism.  The going price for Bitcoin
transactions has been discussed much more thoroughly than the
transaction fee mechanism itself.  For example, the ``blocksize wars''
refers to a bitter dispute within the Bitcoin community over whether
to increase the maximum allowable size of a block (ultimately leading
in~2017 to a split between Bitcoin and a new fork called Bitcoin
Cash)~\cite{blocksize_wars}, and one of the primary arguments put
forth by proponents of larger blocks was that they would prevent (or at
least delay) high transaction fees that would be prohibitive for all
but the most cost-insensitive participants.  Another much-discussed
issue, for example by Carlsten et al.~\cite{ccs16} and Hasu,
Prestwich, and Curtis~\cite{HPC19}, is whether Bitcoin becomes more
vulnerable to various attacks as its block reward decreases (by a
factor of~2 roughly every four years) and the transaction fees per
block increase (as one would expect if the demand for Bitcoin
transactions continues to increase).

Another line of work analyzes Bitcoin transaction fees as a market
equilibrium.  Houy~\cite{houy} and Rizun~\cite{rizun} formalized the
intuitive idea that equilibrium transaction fees should be determined
by the matching of supply with demand.  Richer models of demand, with
waiting costs and pending transactions persisting until inclusion, are
considered by Easley et al.~\cite{EOB17} and Huberman et
al.~\cite{leshno}.  Among other results, these papers show that, as
the fixed block reward decreases, the Bitcoin network can remain
economically viable only if there is sufficient congestion (and
consequent delays) to prop up the market-clearing transaction fees.
More recently, Kiayias et al.~\cite{K+23} explore the related idea of
using delays in a transaction fee mechanism to price differentiate
between latency-sensitive and latency-insensitive transactions.

\vspace{-.1in}

\paragraph{Alternative designs.}
Three previous works focus squarely on the design of alternative
transaction fee mechanisms.  Lavi et al.~\cite{LSZ19} and
Yao~\cite{Y18}, motivated by the aforementioned need to keep Bitcoin
miner revenues high even as the block reward goes to zero, proposed
the ``monopolistic price'' transaction fee mechanism.  In this
mechanism, all transactions included in a block pay the same amount
(per unit size), which is the lowest bid by any included transaction.
(See also Example~\ref{ex:spa_mmic}.)  Miners are then expected to
maximize their revenue (price times quantity), which may involve
restricting the supply (i.e., producing an underfull block) to prop up the price.
Lavi et al.~\cite{LSZ19} and Yao~\cite{Y18} proved that this mechanism
is ``approximately DSIC,'' in the sense that truthful bidding is an
approximately dominant strategy for users as the number of users grows
large.
With respect to the two new notions of
incentive-compatibility introduced in this paper, the mechanism is
MMIC but, on account of choosing revenue-maximization over joint
utility-maximization, is not OCA-proof.
Very recently, Nisan~\cite{N23} studied the price dynamics of this
mechanism over a sequence of blocks in a model with persistent
transactions.

Basu et al.~\cite{beos} also proposed an alternative transaction fee
mechanism to FPAs, with an eye toward stronger incentive-compatibility
guarantees; we summarize a slightly simplified version of it here.
With the monopolistic price mechanism as a starting point, the
mechanism in~\cite{beos} adds two additional ideas.
The first is to automatically charge
only a nominal transaction fee to all transactions in any block that 
is not full.  This rule is intended to prevent miners from
boosting their revenue through the production of underfull blocks,
though by itself the rule is toothless and leads to a
mechanism equivalent to the monopolistic price mechanism (a miner can
costlessly extend its favorite underfull block with minimum bid~$b$ to
a full block with minimum bid~$b$ using fake transactions, all with
bid~$b$).   The second new addition is that transaction fees are partially
paid forward, with the transaction fee revenue from a block~$B$ split
evenly between~$B$'s miner and the miners of the~$\ell-1$ subsequent
blocks (here~$\ell$ is a tunable parameter).
Thus, the miner of a block gets a $1/\ell$ fraction of the transaction
fee revenue in that block, along with a $1/\ell$ fraction of the
combined revenue of the preceding $\ell-1$ blocks.  As a result, for
$\ell \ge 2$, fake transactions now carry a cost: the miner pays their
full transaction fees but recoups only a $1/\ell$ fraction of them as
revenue.  In our terminology, Basu et al.~\cite{beos} prove that their
mechanism is approximately MMIC and approximately DSIC provided 
the range of possible valuations is bounded
and the number of transactions involved is sufficiently large.
Perhaps the biggest vulnerability of
this mechanism is its failure to satisfy OCA-proofness (for every
$\ell \ge 2$): a miner and transaction creators could collude to
move all transaction fees off-chain (paid to the miner), leaving no on-chain fees to pay
forward to future miners (cf., Corollary~\ref{cor:noburn1559_ocaproof}).

\vspace{-.1in}

\paragraph{Credible mechanisms.}
Our notion of incentive-compatibility for myopic miners
(Definition~\ref{def:mmic}) concerns an untrusted auctioneer (the
miner), and as such is related to {\em credible
  mechanisms}~\cite{AL20}.  Intuitively, a mechanism is credible if
the agent tasked with carrying it out has no plausibly deniable
utility-improving deviation.  Interestingly, because miners can
manipulate allocations but not prices
(Remark~\ref{rem:protocol_control}), there is no need to restrict to
``plausibly deniable'' deviations in Definition~\ref{def:mmic}.
Another difference is that the current theory of credible mechanisms,
and in particular the characterizations in~\cite{AL20}, is largely
restricted to single-item auctions (though see~\cite{D+20} for
approximate revenue-optimality results in more general settings).
Blockchain transaction fee mechanisms must work in the more general
setting of (multi-item) knapsack auctions~\cite{AH06}.  Finally, the
model of credible mechanisms developed in~\cite{AL20} assumes that the
auctioneer can communicate privately with each bidder, in sharp
contrast to the publicly visible mempool of pending transactions and
record of confirmed transactions that feature in the model here.
Recent investigations of credible mechanisms in the context of
blockchain protocols and decentralized finance include Ferreira and
Parkes~\cite{FP23} and Chitra, Ferreira, and Kulkarni~\cite{CFK23}.

\vspace{-.1in}

\paragraph{Recent related work.} 
Many papers on transaction fee mechanism design (in addition
to~\cite{CFK23,FP23,K+23,N23}, all cited above) have appeared since
the publication of~\cite{eip1559full} and other preliminary versions
of this work.  Papers that are directly relevant to the specific
mechanism proposed in EIP-1559 include Crapis, Moallemi, and
Wang~\cite{CMW23}, Ferreira et al.~\cite{moroz_aft}, Leonardos et
al.~\cite{monnot_aft,chaos}, and Ndiaye~\cite{abdou}, which investigate
the properties of different base fee update rules; Liu et
al.~\cite{L+22}, Reijsbergen et al.~\cite{reijsbergen2021transaction},
and Zhang and Zhang~\cite{ZZ23}, which provide empirical analyses of the
mechanism after its deployment; and Azouvi et al.~\cite{A+23} and
Hougaard and Pourpouneh~\cite{HP23}, which investigate multi-block
strategies by (non-myopic) miners.  More distantly related are the
works of Canidio~\cite{C23} and Kiayias, Lazos, and
Schlegel~\cite{KLS24}, which explore the pros and cons of burning
transaction fees (as in EIP-1559).

A number of recent papers study the design of transaction fee
mechanisms more generally.  Chung and Shi~\cite{CS23} introduce a
collusion-resistance condition incomparable to OCA-proofness, which
they call ``$c$-side-contract-proofness ($c$-SCP),'' and prove that no
transaction fee mechanism can satisfy both the DSIC and the $c$-SCP
conditions (for any positive integer~$c$).  The $c$-SCP condition
requires that there is never a way for a miner and at most~$c$ users
to collude through an OCA and increase their joint surplus (relative
to the outcome in which those users bid truthfully and the miner
executes the intended allocation rule). OCA-proofness, by contrast,
only concerns OCA-enabled deviations by the ``grand coalition,''
meaning the miner and {\em all} the users with transactions in the
mempool. Intuitively, OCA-proofness protects against the miner and
users banding together to ``cheat the protocol'' (cf.,
Corollaries~\ref{cor:fpaburn_ocaproof}
and~\ref{cor:noburn1559_ocaproof}); the $c$-SCP condition also
protects against a subset of users joining forces with the miner
to cheat the other users. See Gafni and Yaish~\cite{GY22} for further
discussion of this point.  Chung and Shi~\cite{CS23} show that their
impossibility result can be partially circumvented if users have what
they call ``$\gamma$-strict'' utility functions (intuitively, with
each user assuming that overbidding will lead to comeuppance in the
future); see also Tang and Yao~\cite{TY23}.

The proof of the impossibility result in~\cite{CS23} also shows that,
in the ``infinite block size regime'' in which the entire mempool
would fit in a single block, every DSIC and 1-SCP transaction fee
mechanism must burn all user payments (with the miner earning zero
revenue).  
Several recent works have proposed changes to the model with the goal
of evading this negative result:
Shi, Chung, and Wu~\cite{SCW22} propose using a
multi-party computation to generate unmanipulatable randomness;
Wu, Shi, and Chung~\cite{WSC24} assume a lower bound on the
number of non-strategic bidders; and Gafni and Yaish~\cite{GY22} 
restrict attention to a discrete valuation space.

Chen et al.~\cite{C+22}, Gafni and Yaish~\cite{GY22}, and Liu et
al.~\cite{L+24} explore a range of design questions for non-DSIC
transaction fee mechanisms, with an emphasis on Bayesian
incentive-compatible mechanisms.

Finally, Bahrani, Garimidi, and Roughgarden~\cite{BGR23} extend the
model in this paper by allowing a miner to have its own private
valuation for a block (as opposed to caring only about transaction
fees and other in-protocol rewards).  The motivation for this
extension is ``MEV,'' or ``maximal extractable
value''~\cite{flashboys}, which refers to the fact that a miner may be
able to benefit, in a way invisible to the blockchain protocol, by
including or excluding certain transactions.  MEV has become
particularly prominent in the Ethereum ecosystem since the rise of
decentralized finance (``DeFi''), for example with miners frontrunning
trades on decentralized exchanges.  The main impossibility result
in~\cite{BGR23} shows that the presence of MEV makes transaction fee
mechanism design significantly harder: even after setting aside any
collusion-resistance requirements, with private miner valuations, no
transaction fee mechanism can be incentive-compatible simultaneously for
users and for miners.

\section{Preliminaries}\label{s:prelim}

\subsection{Blockchain Transactions}\label{ss:tx}

We consider blockchain protocols that operate in the following way (as
in Bitcoin and Ethereum, for example).  The blockchain protocol
maintains state (such as account balances) and carries out an ordered
sequence of transactions that read from and write to the current state
(such as transfers of the blockchain's native cryptocurrency).  We
assume that each transaction~$t$ has a publicly visible and immutable
{\em size} $s_t$, and that the creator of a transaction is responsible
for specifying a {\em bid} $b_t$ per unit size, indicating a
willingness to pay of up to $s_t \cdot b_t$ (in the native currency)
for the blockchain's execution of their transaction.\footnote{For
  example, in Ethereum, transaction size is called the ``gas limit''
  and is a measure of the cost (in computation, storage, and so on)
  imposed on the system by the transaction's execution.  
  The most basic type of transaction (a simple currency
  transfer) requires 21,000 units of gas; more complex transactions
  require more gas.}

A {\em block} is an ordered sequence of transactions and associated
metadata (such as a reference to the predecessor block).  There is a
cap on the total size of the transactions included in a block, which
we call the {\em maximum block size}.\footnote{For example, in
  Ethereum, prior to the deployment of EIP-1559, the maximum block
  size was~15M gas, enough for roughly 600 of the simplest
  transactions.}  Blocks are created and added to the blockchain by
{\em miners}.  Transactions are submitted by their creators to a
peer-to-peer network; each miner monitors this network, maintains a
{\em mempool} of outstanding transactions, and collects a subset of
them into a block.  To add a block to the blockchain, a miner provides
a proof-of-work in the form of a solution to a computationally
difficult cryptopuzzle; the puzzle difficulty is adjusted over time to
maintain a target rate of block creation (in Ethereum, roughly one
block per 13 seconds).\footnote{As noted in footnote~\ref{foot:pos},
  Ethereum switched from proof-of-work to proof-of-stake on September
  15, 2022. In its proof-of-stake incarnation, miners are replaced by
  validators that have registered a specified amount of native
  currency into a designated staking contract. Every~12 seconds,
  the protocol chooses one validator uniformly at random (using
  pseudorandomness derived from the blockchain's state) as the one
  responsible for assembling the next block.}  Importantly, the miner
of a block has dictatorial control over which outstanding transactions
are included and their ordering within the block.  Transactions are
considered confirmed once they are included in a block that is added
to the blockchain.  The current state of the blockchain protocol is
then the result of executing all the confirmed transactions, in the
specified order.\footnote{Technically, a fork selection rule (e.g.,
  longest-chain of a variant thereof) is used to
  resolve forks, meaning two or more blocks that claim a common
  predecessor. The confirmed transactions are then defined as those
  in the blocks that are well ensconced in the selected chain (that is,
  already extended by sufficiently many subsequent blocks).}

The {\em transaction fee mechanism} is the part of the protocol that
determines the amount that a creator of a confirmed transaction 
pays, and to whom that payment is directed.  Historically, the
biggest blockchains have used a separate first-price (i.e.,
pay-as-bid) auction for each block, with all proceeds going to the
block's miner.

\subsection{The Basic Model}\label{ss:model}

This paper focuses primarily on incentives for miners and users at the
time scale of a single block, and on several important types of
attacks that can be carried out at this time scale (untruthful user
bids, the insertion of fake transactions and other deviations by a
miner, and off-chain agreements between miners and users).  
We discuss incentive
issues and attacks that manifest over longer time scales 
in Section~\ref{s:disc}.

On the supply side, let~$C$ denote the maximum size of a block ($C$ is
for ``capacity'').
On the demand side, we use~$M$ to denote the set of
transactions in a miner's mempool at the time of the current block's
creation.

We associate three parameters with each transaction $t \in M$:
\begin{itemize}

\item a {\em size} $s_t$;

\item a {\em valuation} $v_t$ per unit of size (in the native currency);

\item a {\em bid} $b_t$ per unit of size (in the native currency).

\end{itemize}
The valuation is the maximum per-size price the transaction's
creator would be willing to pay for its execution in the current
block.  The bid corresponds to the per-size price
that the creator actually offers to pay, which in general can be
less (or more) than the valuation.  
The size and bid of a confirmed transaction are recorded
on-chain; the valuation of a transaction is private to its creator.

\subsection{The Design Space: Allocation, Payment, and Burning Rules}\label{ss:model2}

A transaction fee mechanism decides which transactions should be
included in the current block, how much the creators of those
transactions must pay, and to whom their payments are directed.  These
decisions are formalized by three functions: an {\em allocation rule},
a {\em payment rule}, and a {\em burning rule}.  There are two
significant differences between the formalism in this section and that
in classical mechanism design, both dictated by blockchain
idiosyncrasies: payments should depend only on on-chain information
(see Remark~\ref{rem:li}), and revenue can be directed wherever the
protocol sees fit (see Definition~\ref{d:burning}).

\subsubsection{Allocation Rules}

We use $\history = B_1,B_2,\ldots,B_{k-1}$ to denote the sequence of
blocks in the current chain (with~$B_1$ the initial genesis block
and~$B_{k-1}$ the most recently added block) and~$M$ the pending
transactions in the mempool.  (Here $\history$ is for ``history.'')

\begin{definition}[Allocation Rule]
An {\em allocation rule} is a vector-valued
function~$\allocs$ from the on-chain
history~$\history$ and mempool~$M$ to a
0-1 value $\alloct(\history,M)$ for each pending  transaction~$t \in M$.
\end{definition}
A value of~1 for $\alloct(\history,M)$
indicates transaction~$t$'s inclusion in the current
block~$B_k$; a value of~0 indicates its exclusion.
We sometimes write~$B_k = \allocs(\history,M)$, with the
understanding that~$B_k$ is the set of transactions~$t$ for which
$\alloct(\history,M) = 1$.

We consider only feasible allocation rules, meaning allocation rules
that respect the maximum block size~$C$.
\begin{definition}[Feasible Allocation Rule]
An allocation rule $\allocs$ is {\em feasible} if, for every possible
history $\history$ and mempool~$M$,
\begin{equation}\label{eq:feasible}
\sum_{t \in M} s_t \cdot \alloct(\history,M) \le C.
\end{equation}
\end{definition}
We call a set~$T$ of transactions {\em feasible} if they can all be
packed in a single block: $\sum_{t \in T} s_t \le C$.

\begin{remark}[Miners Control Allocations]\label{rem:minerscontrol}
While a transaction fee mechanism is generally designed with a specific
allocation rule in mind, it is important to remember that a miner
ultimately has dictatorial control over the block it creates.
\end{remark}

\subsubsection{Payment and Burning Rules}

The payment rule specifies the revenue earned by the miner from
included transactions.
\begin{definition}[Payment Rule]
A {\em payment rule} is a function~$\prices$ from the current on-chain
history $\history$ and transactions~$B_k$ included in
the current block to a nonnegative number
$\pricet(\history,B_k)$ for each included transaction~$t
\in B_k$.
\end{definition}
The value of $\pricet(\history,B_k)$ indicates the
payment from the creator of an included transaction~$t \in B_k$ to the
miner of the block~$B_k$ (in the native currency, per unit of size).

Finally, the burning rule specifies the amount of money burned
for each of the included transactions.
\begin{definition}[Burning Rule]\label{d:burning}
A {\em burning rule} is a function~$\burns$ from the current on-chain
history~$\history$ and transactions~$B_k$ included in
the current block to a nonnegative number
$\burnt(\history,B_k)$ for each included transaction~$t
\in B_k$.
\end{definition}
The value of $\burnt(\history,B_k)$ indicates the amount of money
burned by the creator of an included transaction~$t \in B_k$ (in the
native currency, per unit of size).  Burning money can be equated with
a lump-sum refund to holders of the currency through deflation, \`{a}
la stock buybacks.  An alternative to money-burning that has similar
game-theoretic properties is to redirect a block's revenue to entities
other than the block's miner, such as a foundation or the miners of
future blocks (see Section~\ref{s:disc} for further discussion).

\begin{example}[First-Price Auction]\label{ex:fpa}
The (intended) allocation rule~$\allocs^{f}$ in the first-price
auctions historically deployed in Bitcoin and Ethereum is to include a
feasible subset of outstanding transactions that
maximizes the sum of the size-weighted bids.
That is, the $\alloct^f$'s are assigned 0-1 values to maximize
\begin{equation}\label{eq:fpa_obj}
\sum_{t \in M} \alloct^{f}(\history,M) \cdot b_t
  \cdot s_t,
\end{equation}
subject to~\eqref{eq:feasible}.\footnote{In practice, some miners
  prefer to employ a greedy heuristic (ordering transactions by bid
  and including the largest feasible prefix of transactions)
rather than solve this knapsack
  problem  optimally.  Because a typical block contains hundreds of
  transactions, the difference in revenue between a greedy and an
  optimal knapsack solution is usually negligible and can be safely
  glossed over.}
A winner then pays its bid (per unit of size), with all revenue going
to the miner (and none burned), 
no matter what the blockchain history and other
included transactions:
$\pricet^f(\history,B_k) = b_t$
and $\burnt^f(\history,B_k) = 0$
for all~$\history$ and~$t \in B_k$.
\end{example}

\begin{remark}[The Protocol Controls Payments and Burns]\label{rem:protocol_control}
A miner does not control the payment or burning rule, except inasmuch as
it controls the allocation, meaning the transactions included
in~$B_k$.  Given a choice of allocation, the on-chain payments and fee
burns are completely specified by the protocol.  (Miners might seek out
off-chain payments, however; see Section~\ref{ss:oca}.)
\end{remark}

\begin{remark}[Mempool-Dependence]\label{rem:li}
The allocation rule~$\allocs$ depends on the mempool~$M$ 
because a miner can base its allocation
decision on the entire set of outstanding transactions.
Payment and burning rules must be computable from the on-chain
history~$\history$, and in particular cannot depend on
outstanding transactions of~$M$ excluded from the current
block~$B_k$.\footnote{In principle, the state of a blockchain
  protocol could keep track of additional data useful for computing a
  payment or burning rule, such as the highest bid by an excluded
  transaction; see Chung and Shi~\cite{CS23} for an in-depth
  exploration of this idea.}
\end{remark}

\begin{definition}[Transaction Fee Mechanism (TFM)]
A {\em transaction fee mechanism} (or {\em TFM}) is a triple $\tfm$ in which
$\allocs$ is a feasible allocation rule, $\prices$ is a payment rule,
and~$\burns$ is a burning rule.
\end{definition}

A TFM is a mechanism for allocating transactions to a single block.  A
blockchain protocol is free to use different TFMs for different
blocks, perhaps informed by the contents of previous blocks.

\section{Miners, Users, and Incentive Compatibility}\label{s:ic}

In a permissionless blockchain protocol such as Bitcoin or Ethereum, a
mechanism designer must guard against harmful deviations from intended
behavior by users (the creators of transactions), by miners, and by
cartels of users and miners.

\subsection{Users}\label{ss:dsic}

We consider a notion of incentive compatibility for users that is
familiar from traditional mechanism design, namely dominant-strategy
incentive compatibility.  Recall from Section~\ref{ss:model} that the
valuation~$\valt$ of a transaction~$t$ is the maximum price (per unit
of size) the transaction's creator would be willing to pay for its
inclusion in the current block.  We assume that a user bids in order
to maximize their net gain (i.e., the value for inclusion minus the
cost for inclusion).  To reason about the different possible bids that
could be attached to a transaction~$t$ submitted to a mempool~$M$, we
use $M(\bidt)$ to denote the result of adding the transaction~$t$ with
bid $\bidt$ to~$M$.  For simplicity, we assume that each transaction
in the current mempool has a distinct creator.\footnote{See
  Section~\ref{s:disc} for further discussion.}

\begin{definition}[User Utility Function]\label{def:userutil}
For a TFM $\tfm$, on-chain history~$\history$, and
mempool~$M$, 
the utility of the originator of a transaction~$t \notin M$ with valuation
$\valt$ and bid~$\bidt$ is
\begin{equation}\label{eq:userutil}
u_t(\bidt) :=
\left( \valt -
\underbrace{p_t(\history,B_k)}_{\substack{\text{payment to miner}\\ \text{(per
    unit size)}}} -
\underbrace{q_t(\history,B_k)}_{\substack{\text{burn}\\ \text{(per unit size)}}} \right)
\cdot s_t
\end{equation}
if~$t$ is included in $B_k = \allocs(\history,M(\bidt))$
and~0 otherwise.
\end{definition}
In~\eqref{eq:userutil}, we highlight the dependence of the utility
function on the argument that is directly under a user's control, the
bid~$\bidt$ submitted with the transaction.  
Because our focus is on incentive issues on the single-block time
scale, we do not explicitly model intertemporal effects such as
waiting costs or otherwise provide additional foundations for the
valuation $\valt$ of immediate inclusion.

We assume that a
transaction creator bids to maximize the utility function
in~\eqref{eq:userutil}.  A TFM is then {\em dominant-strategy
  incentive compatible (DSIC)} if, assuming that the miner carries out
the intended allocation rule, every user 
(no matter what their valuation)
has a dominant strategy---a bid that always maximizes the user's
utility~\eqref{eq:userutil}, no matter what the bids of the other
users.\footnote{To describe some of the transaction fee mechanisms
  used in practice, it will be convenient to allow dominant strategies
  other than truthful bidding in the definition of DSIC.  The
  revelation principle (e.g.,~\cite{M81}) can be
  used to convert any such DSIC mechanism into one in which truthful
  bidding is a dominant strategy.}
FPAs are, of course, not DSIC.
Vickrey-Clarke-Groves (VCG) mechanisms are classical examples of DSIC
mechanisms, with truthful bidding a dominant strategy; as
Example~\ref{ex:spa_mmic} shows, however, these mechanisms are
problematic in a blockchain context.

\subsection{Myopic Miners}

We next formalize incentive compatibility at the single-block time
scale from the perspective of a miner---intuitively, that the miner is
incentivized to implement the intended allocation rule.

We include in our model of miner utility a {\em marginal cost} (per
unit size), denoted by~$\mu$.  (The casual reader is encouraged to
take~$\mu=0$ throughout the paper.)  This parameter reflects the fact
that every transaction included in a block potentially imposes a small
marginal cost on that block's miner.\footnote{For example, in
  proof-of-work blockchain protocols, the probability that a mined
  block is orphaned from the main chain (i.e., the ``uncle rate'')
  appears to increase roughly linearly with the block
  size~\cite{DW13}.}  The parameter~$\mu$ can be interpreted as the
minimum price that a profit-maximizing miner would be willing to
accept in exchange for transaction inclusion when the maximum block
size is not a binding constraint.  For simplicity, we assume
that~$\mu$ is the same for all miners and common knowledge among
users.

\begin{remark}[First-Price Auctions Revisited]\label{rem:fpa_mu}
The
first-price auction in Example~\ref{ex:fpa} is stated for the case
  of $\mu=0$.  More generally, the miner should be expected to maximize
  its revenue minus its costs and the ``$b_t \cdot s_t$'' term
  in~\eqref{eq:fpa_obj} should be
  replaced by $(b_t-\mu) \cdot s_t$:
\begin{equation}\label{eq:fpa_obj2}
\sum_{t \in M} \alloct^{f}(\history,M) \cdot (b_t-\mu)
  \cdot s_t.
\end{equation}
\end{remark}

In addition to choosing an allocation
(Remark~\ref{rem:minerscontrol}), we assume that miners can costlessly
add any number of fake transactions to the mempool (with arbitrary
sizes and bids).  We call a miner {\em myopic} if its utility function
is its net revenue from the current block (given the transactions and
bids submitted by the users).\footnote{Miners may also receive
  transaction-independent rewards from a blockchain protocol for
  producing a block, such as the ``block reward'' in Bitcoin.  Such
  rewards are independent of the miner's actions and therefore
  irrelevant for our game-theoretic analysis.}
\begin{definition}[Myopic Miner Utility Function]\label{def:mmutil}
For a TFM $\tfm$, on-chain history~$\history$, 
mempool~$M$, fake transactions~$F$, and choice~$B_k \subseteq M \cup F$ of
included transactions (real and fake), the utility of a {\em myopic
  miner} is
\begin{equation}\label{eq:mmutil}
u(F,B_k) := 
\underbrace{\sum_{t \in B_k \cap M} \pricet(\history,B_k) \cdot
  s_t}_{\text{miner's revenue}}
-
\underbrace{\sum_{t \in B_k \cap F} \burnt(\history,B_k) \cdot
  s_t}_{\text{fee burn for miner's fake transactions}}
- \underbrace{\mu \sum_{t \in B_k} s_t}_{\text{marginal costs}}.
\end{equation}
\end{definition}
The first term sums over only the real included transactions, as for
fake transactions the payment goes from the miner to itself.  The
second term sums over only the fake transactions, as for real
transactions the burn is paid by their creators (not the miner).
In~\eqref{eq:mmutil}, we highlight the dependence of the utility
function on the two arguments that are under a miner's direct control,
the choices of the fake transactions~$F$ and included (real and 
fake) transactions~$B_k$.\footnote{We can assume that $F \subseteq B_k$,
as there's no point to creating and then excluding a fake
transaction.}

A transaction fee mechanism is generally designed with a specific
allocation rule in mind (Remark~\ref{rem:minerscontrol}), but will
miners actually implement it?
\begin{definition}[Incentive-Compatibility for Myopic Miners
  (MMIC)]\label{def:mmic} \mbox{}\\
A TFM $\tfm$ is {\em incentive-compatible for myopic miners (MMIC)}
if, for every on-chain history~$\history$ and
mempool~$M$, a myopic miner maximizes its utility~\eqref{eq:mmutil} by
creating no fake transactions (i.e., setting $F=\emptyset$) and
following the suggestion of the allocation rule $\allocs$ (i.e.,
setting $B_k = \allocs(\history,M)$).
\end{definition}

For example, FPAs are MMIC---the intended allocation rule maximizes
miner net revenue, which the miner is happy to do (see also Corollary~\ref{cor:mmic}).
Second-price-type auctions are not MMIC, however, as in many cases a
miner can boost its revenue through the inclusion of fake
transactions:
\begin{example}[Second-Price-Type Auctions Are Not MMIC]\label{ex:spa_mmic}
Consider a collection of transactions, all of unit size, and a block that has
room for three of them.  In this setting, a second-price-type auction
would prescribe including the three transactions with the highest bids
and charging each of them the lowest of these three
bids.\footnote{A classical Vickrey auction would prescribe
  charging the highest losing bid rather than the lowest winning bid.
  The former is off-chain and thus unusable in a blockchain
  context, while the latter is on-chain and typically a close enough approximation.}
Now imagine that the top three bids are 10, 8, and 3.  If a miner
honestly executes the auction, its revenue will be $3 \times 3 =
9$.  If the miner instead submits a fake transaction with
bid~8 and executes the auction (with the top two real
transactions included along with the fake transaction), its net
revenue jumps to $2 \times 8 = 16$.
\end{example}

\subsection{Off-Chain Agreements}\label{ss:oca}

Another idiosyncrasy of the blockchain setting is the easy
availability of side channels and the consequent risk of off-chain
collusion by users and miners.  This danger is not
hypothetical for a general smart contracts platform such as Ethereum,
where off-chain markets are already common in practice.\footnote{See,
  e.g., \url{https://docs.flashbots.net/flashbots-auction/overview/}.}

\begin{definition}[Off-Chain Agreement (OCA)]\label{def:oca}
For a miner~$m$ and a set $T$ of transactions,
an {\em off-chain agreement   (OCA)} between 
$T$'s creators and~$m$ specifies:
\begin{itemize}

\item [(i)] a bid vector~$\bids$, with~$\bidt$ indicating the bid
to be submitted with the transaction~$t \in T$;

\item [(ii)] an allocation vector $\allocs$, indicating the
  transactions that the miner~$m$ will include in its block;

\item [(iii)] a per-size transfer~$\tau_t$ from the creator of
  each transaction~$t \in T$ to the miner~$m$.  (If~$\tau_t < 0$,
the transfer should be interpreted as a refund from the
  miner to the transaction creator.)

\end{itemize}
\end{definition}
In an OCA, each creator of a transaction~$t$ agrees to
submit~$t$ with an on-chain bid of $\bidt$ while transferring~$\tau_t
\cdot s_t$ to the miner~$m$ off-chain;
the miner, in turn, agrees to mine a block comprising the
agreed-upon transactions of~$T$.

Intuitively, we define a TFM to be ``OCA-proof'' if no OCA Pareto
improves over a canonical on-chain outcome.
More formally,
by a {\em bidding strategy}, we mean a function $\sigma:\mathbb{R}^{+}
\rightarrow \mathbb{R}^{+}$ mapping user valuations to on-chain bids.
For a valuation profile~$\vals$, $\sigma(\vals)$ denotes the bid
vector obtained by the component-wise application of~$\sigma$.  
With respect to a fixed TFM $\tfm$, a bidding strategy~$\sigma$ is
{\em individually rational} if collective bidding according
to~$\sigma$ guarantees nonnegative utility for all.  Equivalently
(using~\eqref{eq:userutil}), a bidding strategy is individually
rational if
for every on-chain history~$\history$, 
transactions~$T$ with valuations~$\vals$, 
and transaction~$t$
in $B_k=\allocs(\history,T(\sigma(\vals)))$:\footnote{Here $T(\bids)$ denotes the
mempool with transactions specified by~$T$ and bids specified
by~$\bids$.}
\[
\pricet(\history,B_k) +\burnt(\history,B_k) \le \valt.
\]
\begin{definition}[OCA-Proof]\label{def:ocaproof}
A TFM $\tfm$ is {\em OCA-proof} if, for every on-chain
history~$\history$, there exists an individually rational bidding  strategy~$\sigma_{\history}$ such that, 
for every possible set~$T$ of outstanding transactions and
valuations~$\vals$, there is no OCA 
under which
the utility~\eqref{eq:userutil} of every transaction creator and the
utility~\eqref{eq:mmutil} of the miner is strictly higher than in the
outcome $B_k=\allocs(\history,M(\sigma_{\history}(\vals)))$ with
on-chain bids $\sigma_{\history}(\vals)$ and no off-chain
transfers.\footnote{This definition subsumes and corrects the
  preliminary definition of OCA-proofness in~\cite{eip1559full}.}
\end{definition}

In other words, if a TFM is {\em not} OCA-proof, then there is a
possible blockchain history such that, no matter what individually
rational bidding strategy users use, there will be cases in which
off-chain collusion collectively benefits the miner and users.

OCA-proofness can differentiate seemingly similar TFMs.  For example,
we'll see in Section~\ref{ss:ocaproof_results} that first-price
auctions in which all proceeds go to the miner are OCA-proof, while
those in which any amount of revenue is burned are not (intuitively,
because OCAs allow the miner and users to coordinate and evade the
intended burn).

\section{The 1559 and Tipless Mechanisms}\label{s:1559}

This section formalizes the description of the transaction fee
mechanism proposed in EIP-1559, along with an alternative
design that offers a different set of trade-offs (a stronger
incentive-compatibility guarantee for users but weaker resistance to
off-chain agreements).  Both will serve in the next section as
running examples for our main results.

\paragraph{Burning a history-dependent base fee.}
Here are the first three (of eight) key ideas in EIP-1559:
\begin{enumerate}

\item [1.] Each block has a protocol-computed reserve price (per unit size)
 called the {\em base fee}.  Paying the base fee is a
  prerequisite for inclusion in a block.

\item [2.] The base fee is a function of the preceding blocks only, and
does not depend on the transactions included in the current
block.

\item [3.] All revenues from the base fee are burned---that is, permanently
  removed from the circulating supply of the native currency.

\end{enumerate}
The second point is underspecified; how, exactly, is the base fee
derived from the preceding blocks?  Intuitively, increases and
decreases in demand should put upward and downward pressure on the
base fee, respectively.
But the blockchain records only the confirmed
transactions, not the transactions that were priced out.  If miners
publish a sequence of maximum-size blocks, how can the protocol
distinguish whether the current base fee is too low or exactly
right?

\paragraph{Variable-size blocks.}
The next key idea is to relax the constraint that every block has size
at most~$C$ and instead require only that the {\em average} block
size is at most~$C$.
The mechanism in EIP-1559 then uses past block sizes as an on-chain
measure of demand, with big blocks (size above~$C$) and small
blocks (size less than~$C$) signaling positive
and negative excess demand, respectively.
Some finite maximum block size is still needed to control network
congestion, which in EIP-1559 is twice the average block size:
\begin{enumerate}

\item [4.] Double the maximum block size (i.e., define $C_{max}:=2C$),
with the old maximum~$C$ serving as the new {\em target} block size
($C_{target} := C$).

\item [5.] Adjust the base fee upward or downward whenever the
  size of the latest block is bigger or smaller than the target
  block size, respectively.\footnote{Precisely, empty and
    maximum-size blocks decrease and increase the base fee by 12.5\%,
    respectively, with the effect of other block sizes then determined
    by linear interpolation.  (In particular, blocks matching the
    target size do not alter the base fee.)}
\end{enumerate}

If the base fee is burned rather than given to miners, why should
miners bother to include any transactions in their blocks at all?
Also, what happens when there are lots of transactions (with total
size exceeding~$C_{max}$)
willing to pay the current base fee?  

\paragraph{Tips.}
The transaction fee mechanism proposed in EIP-1559 addresses the
preceding two questions by allowing the creator of a transaction to
specify a {\em tip}, to be paid above and beyond the base fee, which
is transferred to the miner of the block that includes the transaction
(as in a first-price auction).  
Intuitively, small tips should be sufficient to
incentivize a miner to include a transaction during a period of stable
demand, when there is room in the current block for all the
outstanding transactions that are willing to pay the base fee.  Large
tips can be used to encourage special treatment of a transaction, such
as the immediate inclusion in a block in the midst of 
a sudden demand spike.  The final ingredients of the mechanism in
EIP-1559 are:
\begin{enumerate}

\item [6.] Rather than a single bid, the creator of a transaction is
now  responsible for specifying both
a {\em tip} and a {\em fee cap} for it.
A transaction will be included in a block only if its fee cap is at
least the block's base fee.

\item [7.] If a size-$s$ transaction with tip~$\delta$ and fee cap $c$
  is included in a block with base fee~$r$, the transaction
  creator pays a total of $s \cdot \min\{ r+\delta, c\}$.

\item [8.] Revenue from the base fee (that is, $s \cdot r$) is burned
  and the remainder ($s \cdot \min\{ \delta, c-r\}$) is transferred to
  the miner of the block.

\end{enumerate}
Thus, with respect to a base fee~$r$, a transaction~$t$ with
tip~$\delta_t$ and fee cap~$c_t$ is interpreted as a transaction with
bid $\bidt = \min \{ r + \delta_t, c_t\}$.\footnote{%
  Our model and results consider only the single-block time scale,
  within which the base fee of this mechanism would be fixed and
  public.  As such, the components of the mechanism that govern how
  the base fee evolves over multiple blocks, including the
  relationship between~$C_{target}$ and~$C_{max}$ and the exact
  formula for the base fee update rule, are not directly relevant for
  our analysis. However, the interpretation of
  Theorem~\ref{t:1559dsic} and Corollary~\ref{cor:tipless_ocaproof}
  as ``usually DSIC'' and ``usually OCA-proof'' guarantees 
  does implicitly assume a base fee update rule 
  that, under ``normal conditions,'' prevents the base fee from
  becoming excessively low relative to the current demand
  (in the sense of Definition~\ref{def:low}); see also footnote~\ref{foot:alpha}.}

\paragraph{The 1559 mechanism.}
We are now in a position to phrase EIP-1559's transaction fee
mechanism---for shorthand, the {\em 1559 mechanism}---in 
the language of Section~\ref{s:prelim}.  
Because the base fee of a block depends solely on the contents of the
preceding blocks, we can denote by~$\alpha(\history)$ 
the base fee of a block~$B_k$ with history~$\history$.
\begin{definition}[1559 Mechanism]\label{d:1559}
For each history~$\history$ and corresponding base fee
$r = \alpha(\history)$:
\begin{itemize}

\item [(a)] the (intended) allocation rule~$\allocs^{\nine}$ of the
  1559  mechanism 
is to include a feasible subset of outstanding transactions that
maximizes the sum of the size-weighted bids,
less the cost and total base fee paid (and 
subject to the block size constraint~\eqref{eq:feasible}, with
capacity~$C_{max}$):
\begin{equation}\label{eq:1559_obj}
\sum_{t \in M \,:\, \bidt \ge r} \alloct^{\nine}(\history,M) \cdot
(\bidt - r - \mu)   \cdot s_t;
\end{equation}

\item [(b)] the payment rule of the 1559 mechanism is
\[
\pricet^{\nine}(\history,B_k) = b_t - r
\]
for all $t \in B_k$;

\item [(c)] the burning rule of the 1559 mechanism is 
\[
\burnt^{\nine}(\history,B_k) = r
\]
for all $t \in B_k$.

\end{itemize}
\end{definition}

\paragraph{The tipless mechanism.}
We next define the {\em tipless mechanism}, so called because it is
essentially the 1559 mechanism with constant and hard-coded tips
rather than variable and user-specified tips.  As with the 1559
mechanism, each block has a base fee $r =
\alpha(\history)$ that depends on past blocks and is
burned.
The creator of a transaction~$t$ specifies a fee cap~$c_t$
but no tip.  This parameter induces a bid $\bidt$ for the transaction
with respect to any given base fee~$r$, namely
$\bidt = \min\{ r + \delta, c_t \}$.  
Here $\delta$ is a hard-coded tip to incentivize miners to include
transactions---for example, equal to (or perhaps slightly higher than)
the marginal cost~$\mu$.\footnote{More
  generally, the hard-coded tip~$\delta$ could be adjusted over time 
via hard forks, as is typically done for a number of other protocol
parameters.}
In effect, the bid space of the mechanism is~$[0,r+\delta]$, with
higher bids automatically interpreted as $r+\delta$ by the protocol.
Only transactions with bid $r+\delta$ are eligible for inclusion in a
block with base fee~$r$; transactions with lower bids included in the
block are considered invalid by the protocol.

\begin{definition}[Tipless Mechanism]\label{d:tipless}
Fix a hard-coded user tip~$\delta$.
For each history~$\history$ and corresponding base fee
$r = \alpha(\history)$:
\begin{itemize}

\item [(a)] the (intended) allocation rule~$\allocs^{\delta}$ of the
tipless mechanism is to maximize miner revenue from eligible
transactions (i.e., those with bid at least $r+\delta$), less costs and
subject to~\eqref{eq:feasible} with maximum block size~$C_{max}$:
\begin{equation}\label{eq:tipless_obj}
\sum_{t \in M \,:\, \bidt \ge r+\delta} \alloct^{\delta}(\history,M) \cdot
(\delta - \mu) \cdot s_t,
\end{equation}
or equivalently, in the intended regime with $\delta \ge \mu$, to
include a largest-possible subset of eligible
transactions;\footnote{Ties between subsets are intended to be broken
  consistently and independently of transactions' fee caps.}

\item [(b)] the payment rule of the tipless mechanism is
\[
\pricet^{\delta}(\history,B_k) = \delta
\]
for all $t \in B_k$;

\item [(c)] the burning rule of the tipless mechanism is 
\[
\burnt^{\delta}(\history,B_k) = r
\]
for all $t \in B_k$.

\end{itemize}
\end{definition}

\section{Main Results: Which TFMs Are MMIC, DSIC, or OCA-Proof?}\label{s:results}

This section develops general tools for reasoning about the incentive
guarantees of different transaction fee mechanisms.  We use six
specific TFMs to illustrate our results:
\begin{enumerate}

\item A first-price auction (FPA), as described in
  Example~\ref{ex:fpa}.

\item A second-price-type auction (SPA), similar to
  Example~\ref{ex:spa_mmic}.  
For concreteness, we assume that the intention is for a miner to order
the outstanding transactions in nonincreasing order of bid (per unit
size) and include the largest feasible prefix of transactions.  All
included transactions pay the lowest accepted bid (per unit size),
and all revenue is passed on to the miner.

\item A first-price auction in which a~$\beta \in (0,1]$ fraction of
  the transaction fees are burned ({\em $\beta$-burn FPA})---that is, with
  $\pricet(\history,B_k) = (1-\beta)b_t$
and $\burnt(\history,B_k) = \beta b_t$ for an included
transaction with bid~$\bidt$.

\item The 1559 mechanism, as described in Definition~\ref{d:1559}.

\item The {\em $\beta$-burn 1559 mechanism}, in which a $\beta \in
  [0,1)$ fraction of the base fee revenues are burned and 
the rest
  are passed on to a block's miner---that is, with
  $\pricet(\history,B_k) = b_t - \beta r$
and $\burnt(\history,B_k) = \beta r$ for an included
transaction with bid~$\bidt$ (where~$r=\alpha(\history)$).
The intended allocation rule is analogous to that of the 1559
mechanism, with transactions chosen to maximize
\[
\sum_{t \in M \,:\, \bidt \ge r} \alloct(\history,M) \cdot
(\bidt - \beta r - \mu)   \cdot s_t.
\]

\item The tipless mechanism, as described in Definition~\ref{d:tipless}.

\end{enumerate}
Table~\ref{table:main} summarizes the implications of our results for
these six designs.

\begin{table}[h]
\begin{center}
\vspace{.5\baselineskip}
\begin{tabular}{|c|c|c|c|}\hline
TFM & MMIC? & DSIC? & OCA-proof?  \\ \hline
FPA & yes (Cor.~\ref{cor:mmic}) & no (obvious) & yes (Cor.~\ref{cor:fpa_ocaproof})\\
SPA & no (Ex.~\ref{ex:spa_mmic}) 
& almost\tablefootnote{An SPA
                                        that uses the lowest included
                                        bid as a proxy for the highest
                                        excluded bid is not generally
                                        DSIC.  However, it 
                                        is approximately DSIC (with
                                        truthful bidding always nearly
                                        maximizing bidder utility)
                                        whenever these two values are
                                        close (as one would expect in
                                        a block with hundreds of transactions).}
& almost (Rem.~\ref{rem:spa_ocaproof})\\
$\beta$-burn FPA & yes (Cor.~\ref{cor:mmic}) & no (obvious) & no (Cor.~\ref{cor:fpaburn_ocaproof})\\
1559 & yes  (Cor.~\ref{cor:mmic}) & usually (Thm.~\ref{t:1559dsic}) & yes (Cor.~\ref{cor:1559ocaproof})\\
$\beta$-burn 1559 & yes  (Cor.~\ref{cor:mmic}) & usually (Rem.~\ref{rem:1559dsic}) & no (Cor.~\ref{cor:noburn1559_ocaproof}))\\
tipless & yes (Cor.~\ref{cor:mmic}) & yes (Thm.~\ref{t:tipless_dsic}) & usually (Cor.~\ref{cor:tipless_ocaproof}+Rem.~\ref{rem:tipless_ocaproof})\\
\hline
\end{tabular}
\caption{Which of the six listed TFMs are MMIC, DSIC, or OCA-proof.}\label{table:main}
\end{center}
\end{table}

Thus, if we assess these TFMs solely according to these three types of
incentive guarantees, FPAs dominate $\beta$-burn FPAs with $\beta >
0$; the 1559
mechanism dominates these and $\beta$-burn 1559 mechanisms with $\beta
< 1$; and the 1559 mechanism and the tipless mechanism are incomparable.

\subsection{MMIC and Non-MMIC TFMs}\label{ss:mmic_results}

The MMIC condition (Definition~\ref{def:mmic}) states that a
revenue-maximizing miner should be incentivized to follow the intended
allocation rule.  Example~\ref{ex:spa_mmic} shows that not all
interesting mechanisms are MMIC, and in particular that
second-price-type auctions do not satisfy the condition.  What goes
wrong in Example~\ref{ex:spa_mmic} is that the payment collected from
one transaction depends on the other included transactions.  We call a
payment rule {\em separable} if this is not the case.
\begin{definition}[Separable Payment Rule]\label{d:separable}
A payment rule~$\prices$ is {\em separable} if, for every
on-chain history~$\history$ and block~$B_k$, 
the payment $\pricet(\history,B_k)$
of an included
transaction~$t \in B_k$ is independent of the set~$B_k - \{t\}$ of
other included transactions and their bids.
\end{definition}

For a separable payment rule~$\prices$
and a fixed transaction~$t$ (with some bid $\bidt$),
we can write~$\pricet(\history)$ for the payment that~$t$'s
creator would pay should~$t$ be included in the block~$B_k$ that
follows the history~$\history$.
(This notation is well defined by separability.)

A separable payment rule~$\prices$ suggests a corresponding {\em
  revenue-maximizing} allocation rule~$\allocs$, in which a miner
always chooses the most profitable subset of transactions.
That is, given on-chain history~$\history$ and a
mempool~$M$, the $\alloct$'s are assigned 0-1 values to maximize
\begin{equation}\label{eq:rm}
\sum_{t \in M} \alloct(\history,M) \cdot
(\pricet(\history) - \mu) \cdot s_t,
\end{equation}
subject to feasibility.

Every TFM that uses a separable payment rule and the corresponding
revenue-maximizing allocation rule is MMIC.
\begin{theorem}[Separable Payments and Revenue Maximization Imply MMIC]\label{t:mmic}
\mbox{}\\
If $\prices$ is a separable payment rule, $\allocs$ is the
corresponding revenue-maximizing allocation rule, and $\burns$ is an
arbitrary burning rule, then the TFM $\tfm$ is MMIC.
\end{theorem}

\begin{proof}
Fix an on-chain history~$\history$, a mempool~$M$, and a
marginal cost~$\mu \ge 0$.
By Definition~\ref{d:1559},
myopic miner utility~\eqref{eq:mmutil} equals
\begin{equation}\label{eq:mmic_proof}
u(F,B_k) := 
\underbrace{\sum_{t \in B_k \cap M} (\pricet(\history) -\mu) \cdot
  s_t}_{\text{revenue less marginal costs}}
-
\underbrace{\sum_{t \in B_k \cap F} (\mu +
  \burnt(\history,B_k)) \cdot
  s_t}_{\text{fake transaction costs}},
\end{equation}
where~$B_k$ denotes the transactions included by the miner and~$F$ the
fake transactions that it creates.  Included fake transactions can only
increase the second term (as $\mu$ and $\burns$ are nonnegative)
while leaving the
first unaffected (because~$\prices$ is separable), and so a myopic
miner can be assumed to include only real transactions in~$B_k$.  In
this case, myopic miner utility equals
\[
\sum_{t \in B_k} (\pricet(\history)-\mu) \cdot s_t,
\]
which is identical to the quantity~\eqref{eq:rm} maximized by
the revenue-maximizing allocation rule.  Thus, myopic miner utility is
maximized by following the allocation rule and setting $B_k$ equal to $\allocs(\history,M)$.
\end{proof}

FPAs use the separable payment rule with $\pricet(\history) = \bidt$
and the corresponding revenue-maximizing rule (see~\eqref{eq:fpa_obj2}).
$\beta$-burn FPAs use the separable rule $\pricet(\history) =
(1-\beta)\bidt$ and the same revenue-maximizing allocation rule.
The 1559, $\beta$-burn 1559, and tipless mechanisms use the separable
payment rules given by 
$\pricet(\history) = \bidt-r$,
$\pricet(\history) = \bidt-\beta r$,
and $\pricet(\history) = \delta$, respectively, where~$r$ denotes the
1559 mechanism's current base fee (which, crucially, depends only
on~$\history$) and~$\delta$ is the hard-coded tip in the tipless mechanism.
By definition, all three mechanisms use the corresponding
revenue-maximizing allocation rules.
Applying Theorem~\ref{t:mmic}:
\begin{corollary}[Five MMIC TFMs]\label{cor:mmic}
FPAs, $\beta$-burn FPAs, the 1559 mechanism, the $\beta$-burn 1559
mechanism, and the tipless mechanism are all MMIC.
\end{corollary}

\begin{remark}[Non-Separable MMIC TFMs]
There exist MMIC TFMs that do not employ a separable payment rule.
For example, define~$\burns$ as the all-zero function,
$\pricet(\history,B_k) = \bidt$ if~$t$ is the highest-bidding
transaction included in~$B_k$ (with ties broken arbitrarily but
consistently), and $\pricet(\history,B_k) = 0$ otherwise.
Define~$\allocs$ as the rule that includes only the highest-bidding
transaction in the mempool (or, if every bid is less than~$\mu$,
includes nothing).  This TFM is MMIC even though~$\prices$ is not
separable.  The monopolistic price mechanism proposed by Lavi et
al.~\cite{LSZ19} and discussed in Section~\ref{s:intro} is another
example. 
Characterizing the MMIC TFMs is an interesting open
research question.
\end{remark}

\subsection{DSIC and Non-DSIC TFMs}\label{ss:dsic_results}

The DSIC condition (Section~\ref{ss:dsic}) states that every
transaction creator should always have a dominant bidding strategy---a
bid that maximizes their utility~\eqref{eq:userutil}, no matter
what the bids of others.
The optimal bid in an FPA or $\beta$-burn FPA
depends on others' bids, so these TFMs are not DSIC.

The tipless mechanism is an example of a DSIC TFM.
\begin{theorem}\label{t:tipless_dsic}
The tipless mechanism is DSIC.
\end{theorem}

\begin{proof}
Fix an on-chain history $\history$ and corresponding
base fee $r = \alpha(\history)$.
The claim is that the bidding strategy
$\sigma(\valt) = \min\{ r+\delta, \valt \}$
is a dominant strategy for every bidder,
where~$\delta$ denotes the value of the hard-coded tip in the tipless
mechanism.\footnote{In the mechanism's implementation
  (Section~\ref{s:1559}), this is
  precisely the bid induced by a truthfully reported fee cap (with
  $c_t = \valt$) and the current base fee~$r$.}

Fix a set~$T$ of transactions with valuations~$\vals$ and a
transaction~$t \in T$.
Suppose $t$'s creator bids
$\bidt = \sigma(\valt) = \min\{ r+\delta, \valt \}$.
If $t$ is a low-value transaction (with $\valt < r + \delta$), every
alternative bid $\hat{\bid}_t$ either has no effect on $t$'s utility or
leads to~$t$'s inclusion in the block; the latter occurs only when
$\hat{\bid}_t \ge r+\delta$, in which case the creator's utility drops
from~0 to $(\valt -\hat{\bid}_t) \cdot s_t < 0$.  
For a high-value transaction (with $\valt \ge r + \delta$), every
alternative  bid $\hat{\bid}_t$ either has no effect on the creator's
utility or, if the alternative bid triggers~$t$'s exclusion, drops its
utility from a nonnegative number $(\valt - r- \delta) \cdot s_t \ge
0$ to~0.\footnote{Recall that the tipless mechanism's allocation rule
  includes a
largest-possible subset of the eligible transactions (i.e.,
transactions~$t$ with $\bidt \ge r +\delta$), breaking ties
between subsets in a consistent and bid-independent way.  (Unless
$\delta < \mu$, in which case the mechanism never includes any
transactions and is trivially DSIC.)}
We conclude that the bid $\sigma(\valt) = \min\{ r+\delta, \valt \}$
is always utility-maximizing for~$t$'s creator.
\end{proof}

The 1559 mechanism is not DSIC in general, as for the special case of
a zero base fee it is equivalent to an FPA.  
However, given that the base fee is automatically adjusted over time
in response to excess demand, one might expect that, in a typical
block, the base fee is sufficiently high to exclude all but a
reasonable number of transactions in the mempool.
Can we at least argue that, if the base fee ``does its
job,'' then the 1559 mechanism is DSIC?\footnote{A natural conjecture
  is that, for an appropriately tuned base fee update rule~$\alpha$,
  excessively low base fees should arise only in short transitory
  periods while waiting for the base fee to catch up to a large and
  sudden demand spike.  Preliminary investigations by Monnot~\cite{monnot}
  with synthetic data provide some initial support for this conjecture.\label{foot:alpha}}

\begin{definition}[Excessively Low Base Fee]\label{def:low}
Let~$\mu$ denote the marginal cost.
In the 1559 mechanism, a base fee~$r$ is {\em excessively low} for a
set~$T$ of transactions with valuations $\vals$ if the demand at
price~$r+\mu$ exceeds the maximum block size~$C_{max}$:
\begin{equation}\label{eq:low}
\underbrace{\sum_{t \in M \,:\, \valt \ge r+\mu} s_t}_{\text{demand at
    price $r+\mu$}} > C_{max}.
\end{equation}
\end{definition}

Are excessively low base fees the only obstruction to
the DSIC property?  Not quite.  The issue is that if, for
whatever reason, users choose to overbid (with $\bidt > \valt$), a
base fee may act as if it is excessively low (with respect to the
reported bids) even though it is not (with respect to the true
valuations).
The next result proves that these are the only two
obstructions to achieving DSIC---without them, 
the 1559 mechanism effectively acts as a posted price
mechanism with a price (per unit size) equal to the base fee~$r$
plus the miner marginal cost~$\mu$.

\begin{theorem}[The 1559 Mechanism Is Usually DSIC]\label{t:1559dsic}
Fix an on-chain history $\history$ and corresponding
base fee $r = \alpha(\history)$, a marginal cost~$\mu$,
and a set~$T$ of transactions with valuations~$\vals$.
If~$r$ is not excessively low for~$T$ and transaction creators cannot
overbid, the bidding strategy
$\bidt = \sigma(\valt) = \min\{ r+ \mu, \valt \}$
is a dominant strategy for every bidder.\footnote{In the mechanism's implementation
  (Section~\ref{s:1559}), this is
  precisely the bid induced by a truthfully reported fee cap (with
  $c_t = \valt$), a tip of~$\mu$, and a current base fee of~$r$.}
\end{theorem}

\begin{proof}
Fix a transaction~$t \in T$, with valuation~$\valt$.
Suppose first that $\valt < r + \mu$.
The objective~\eqref{eq:1559_obj} of the 1559 allocation rule
prescribes including only transactions~$s \in T$ with $\bid_s
\ge r + \mu$.  
If~$t$'s creator bids $\bidt = \min\{ r+ \mu, \valt \}
= \valt$, the transaction will be excluded from the block and the
resulting utility will be~0.  Every alternative bid~$\hat{\bid}_t$ either
leads to the same outcome or, if it results in $t$'s inclusion in the
block, leads to negative utility (at most $\valt - (r+\mu) < 0$).

Now suppose that $\valt \ge r + \mu$.
Because transaction creators cannot overbid,
the transactions~$w \in T$ with $\bid_w \ge r + \mu$
are a subset of the transactions~$w \in T$ with $\val_w \ge r + \mu$.
Thus, because~$r$ is not excessively low for~$T$, there is room for
all of these transactions (no matter what~$\bidt$ is):
\begin{equation}\label{eq:1559uic}
\underbrace{\sum_{w \in T \,:\, \bid_w \ge r + \mu}
  s_w}_{\text{total size of included txs}}
\quad \le 
\underbrace{\sum_{w \in T \,:\, \val_w \ge r + \mu}
  s_w}_{\text{demand at price $r + \mu$}} \le \quad C_{max}.
\end{equation}
If~$t$'s creator bids $\bidt = \min\{ r+ \mu, \valt \}
= r+\mu$, the transaction will be included in the block and the
resulting utility will be~$\valt - (r+\mu) \ge 0$.  Every alternative
bid either leads to~$t$'s exclusion (resulting in utility~0) or
to~$t$'s inclusion at a price higher than~$r+\mu$.
We conclude that there is no alternative bid for~$t$ that increases
its creator's utility, and hence $\sigma(\valt) = \min\{ r+ \mu, \valt
\}$ is a dominant bidding strategy.
\end{proof}

\begin{remark}[Alternative Interpretation of Theorem~\ref{t:1559dsic}]
The proof of Theorem~\ref{t:1559dsic} shows that, whenever the base
fee is not excessively low and all other transaction creators do not
overbid (e.g., on account of overbidding being a dominated strategy),
a best response of a transaction creator with valuation
$\valt$ is to bid $\sigma(\valt) =\min\{r+\mu,\valt\}$. 
Thus, when the base fee is not excessively low, the outcome in
which all transaction creators bid according to the strategy~$\sigma$
constitutes an ex post Nash equilibrium.
\end{remark}

\begin{remark}[The $\beta$-Burn 1559 Mechanism]\label{rem:1559dsic}
The $\beta$-burn 1559 and 1559 mechanisms make the same allocation
decisions and charge the same total price (i.e., miner payment plus
burn) for included transactions, and differ only in how the payments
by users are directed.  Thus the two mechanisms are identical from the
user perspective, and Theorem~\ref{t:1559dsic} carries over
immediately to all $\beta$-burn 1559 mechanisms.
\end{remark}

\subsection{OCA-Proof and Non-OCA-Proof TFMs}\label{ss:ocaproof_results}

The OCA-proofness condition (Definition~\ref{def:ocaproof}) requires
the existence of a canonical and individually rational on-chain
outcome that cannot be Pareto improved by any off-chain agreement
(specifying the on-chain allocation and bids, and the off-chain
transfers between users and the block's miner).

Because OCAs can specify arbitrary transfers, we can characterize
OCA-proofness in terms of a surplus-maximization property.  The next
definition is the sum of the utility functions of the miner and all
the creators of pending transactions.
\begin{definition}[Joint Utility]
For an on-chain history~$\history$ and mempool~$M$, 
the {\em joint utility} of the miner and the creators of transactions
in~$M$ for a block~$B_k$ is
\begin{equation}\label{eq:jointutil}
\sum_{t \in B_k} \left( \valt - q_t(\history,B_k) - \mu
\right) \cdot s_t.
\end{equation}
\end{definition}

From the perspective of a coalition of users and a miner, on-chain and
off-chain payments from the users to the miner (the $\pricet$'s
specified by the TFM and the $\tau_t$'s specified by the OCA) remain
within the coalition and cancel out; burned money (the
$\burnt$'s) is transferred outside the coalition and
is therefore a loss.  

\begin{prop}[OCA-Proof $\Leftrightarrow$ Joint Utility-Maximization]\label{prop:ocaproof}
A TFM $\tfm$ is\\ OCA-proof if and only if, for every on-chain
history~$\history$, there exists an individually rational bidding
strategy~$\sigma_{\history}$ such that, 
for every possible set~$T$ of outstanding transactions and
valuations~$\vals$, 
the outcome $B_k=\allocs(\history,T(\sigma_{\history}(\vals)))$ 
maximizes the joint utility~\eqref{eq:jointutil} over every possible
on-chain outcome $\allocs(\history,T(\bids))$.
\end{prop}

\begin{proof}
For the ``only if'' direction, suppose 
there exists a history~$\history$ such that, for every individually
rational bidding strategy~$\sigma_{\history}$, there is a set~$T$ of transactions
with valuations~$\vals$ and a set~$\bids$
of on-chain bids such that the joint utility
of the outcome~$\allocs(\history,T(\bids))$ is strictly larger than
that of the outcome~$\allocs(\history,T(\sigma_{\history}(\vals)))$.  
Then, there is an OCA that uses on-chain bids~$\bids$ and suitable
transfers to share the additional joint utility (relative to the
outcome~$\allocs(\history,T(\sigma_{\history}(\vals)))$) such that all
transaction creators and the miner are strictly better off.

Conversely, suppose that for every history~$\history$, there exists an
individually rational bidding strategy~$\sigma_{\history}$ such that,
for every set~$T$ of transactions and valuation vector~$\vals$, the
bid profile~$\sigma_{\history}(\vals)$ maximizes the joint utility
over all possible on-chain bids $\bids$.  Then, for every OCA, its
joint utility is at most that of
$\allocs(\history,T(\sigma_{\history}(\vals)))$; because the joint
utility is the sum of the miner and users' utility functions, some
participant's utility under the OCA is at most that
in~$\allocs(\history,T(\sigma_{\history}(\vals)))$.  Thus, no OCA Pareto
improves over the outcome with on-chain
bids~$\sigma_{\history}(\vals)$ and no off-chain transfers.
\end{proof}

Proposition~\ref{prop:ocaproof} reduces the task of verifying
OCA-proofness to checking whether or not there is a joint
utility-maximizing and individually rational bidding strategy.
We proceed to check each of our six running examples.

\begin{corollary}\label{cor:fpa_ocaproof}
FPAs are OCA-proof.
\end{corollary}

\begin{proof}
Fix an (irrelevant) history $\history$.
Let~$\gamma \in (0,1]$ be arbitrary, and 
define a bidding strategy~$\sigma$ by
\[
\sigma(\valt) = \min\{ \valt, \mu + \gamma \left( \valt - \mu \right) \}.
\]
For a first-price auction, this bidding strategy is individually
rational.
Consider a set~$T$ of transactions with valuations~$\vals$.
Because the burning rule $\burns$ is the all-zero function, the
objective~\eqref{eq:fpa_obj2} maximized by the allocation
rule~$\allocs^f$ with bids $\sigma(\vals)$ is identical (modulo the
scaling factor~$\gamma$) to the joint utility~\eqref{eq:jointutil}.
Thus, the joint utility of the on-chain outcome with bids
$\sigma(\vals)$ cannot be improved upon by any OCA.  
\end{proof}

\begin{remark}[SPAs Are Almost OCA-Proof]\label{rem:spa_ocaproof}
In an SPA,
the burning rule $\burns$ is the all-zero function and the
joint utility~\eqref{eq:jointutil} reduces to the social welfare
$\sum_{t \in B_k} (\valt - \mu)\cdot s_t$.  
Had we defined an SPA using the welfare-maximizing allocation rule,
it would be an OCA-proof transaction fee mechanism (by
Proposition~\ref{prop:ocaproof}, using the identity bidding strategy~$\sigma(\valt)=\valt$).  
For our definition of SPAs based on a greedy heuristic allocation rule,
the outcome may have slightly less than the maximum-possible social
welfare, and so (by Proposition~\ref{prop:ocaproof}) there may be an
opportunity for a (small) Pareto improvement.
\end{remark}

The proof that FPAs are OCA-proof
(Corollary~\ref{cor:fpa_ocaproof}) can be extended to the 1559
mechanism.

\begin{corollary}\label{cor:1559ocaproof}
The 1559 mechanism is OCA-proof.
\end{corollary}

\begin{proof}
Fix a history $\history$ and corresponding base fee
$r=\alpha(\history)$.
Let~$\gamma \in (0,1]$ be arbitrary, and 
define a bidding strategy~$\sigma$ by
\[
\sigma(\valt) = \min\{ \valt, \mu + r + \gamma \left( \valt - \mu - r \right) \}.
\]
This bidding strategy is individually rational.
Consider a set~$T$ of transactions with valuations~$\vals$.
The
objective~\eqref{eq:1559_obj} maximized by the allocation
rule~$\allocs^{\nine}$ with bids $\sigma(\vals)$ is identical (modulo
the scaling factor~$\gamma$) to the joint utility~\eqref{eq:jointutil}.
Thus, the joint utility of the on-chain outcome with bids
$\sigma(\vals)$ cannot be improved upon by any OCA, and
Proposition~\ref{prop:ocaproof} then implies that the mechanism is
OCA-proof.
\end{proof}

The tipless mechanism fails OCA-proofness in the same regime in which
the 1559 mechanism loses its DSIC property, the regime of an
excessively low base fee (Definition~\ref{def:low}).
In effect, the tipless mechanism retains DSIC in this regime by
disallowing bidders to differentiate themselves through high bids, and
fails OCA-proofness for the same reason.
\begin{corollary}[The Tipless Mechanism Is Usually OCA-Proof]\label{cor:tipless_ocaproof}
Fix an on-chain\\ history $\history$ and corresponding
base fee $r = \alpha(\history)$, and a set~$T$ of
transactions with valuations~$\vals$ such that~$r$ is not excessively
low.
The tipless mechanism, with hard-coded tip~$\delta$ equal to the
marginal cost~$\mu$, is OCA-proof.
\end{corollary}

\begin{proof}
Define an (individually rational) bidding strategy~$\sigma$ by
$\sigma(\valt)  = \min \{ r + \delta, \valt \} = \min \{ r + \mu,
\valt \}$.
The joint utility~\eqref{eq:jointutil} of the miner and users for the
outcome~$B_k$ is
\begin{equation}\label{eq:tipless}
\sum_{t \in B_k} ( \valt - r - \mu) \cdot s_t.
\end{equation}
Because~$r$ is not excessively low for~$T$, the total size of the
transactions~$t$ with $\valt \ge r + \mu = r+\delta$ is at most the
maximum block size~$C_{max}$.  The joint utility~\eqref{eq:tipless} is
therefore maximized by including these (and only these) transactions,
which is precisely the outcome under the bids~$\sigma(\vals)$.
Proposition~\ref{prop:ocaproof} then implies that, 
under the assumption of a base fee that is not excessively low,
the tipless mechanism is OCA-proof.
\end{proof}

\begin{remark}[The Tipless Mechanism Is Not Generally OCA-Proof]\label{rem:tipless_ocaproof}
The tipless mechanism is not generally OCA-proof when the base fee~$r$
is excessively low (even with $\delta=\mu$).  For consider an
arbitrary individually rational bidding strategy~$\sigma$.  If
$\sigma(v) < r+\delta$ for
some~$v > r+\delta$, a collection of transaction creators all with
valuation~$v$ would be better off in an OCA with a miner than
bidding~$\sigma(v)$ on-chain (which would lead to the automatic
exclusion of all the transactions).  On the other hand, if $\sigma(v) \ge r+\delta$ for all
$v > r+\delta$, consider a collection of transactions~$T$ with total
size~$\sum_{t \in T} s_t$ bigger than the maximum block size~$C_{max}$
and with creators that have
distinct valuations, all bigger than $r+\delta$.  The intended
allocation rule then instructs the miner to include a subset of
transactions with the maximum-possible total size (subject to the
block capacity~$C_{max}$).  This will not generally be a
subset~$S \subseteq T$ of transactions that maximizes the joint
utility $\sum_{t \in S} (v_t-r-\mu) \cdot s_t$ subject to the block
capacity; by Proposition~\ref{prop:ocaproof}, such an outcome can be
Pareto improved through an OCA.

The tipless mechanism is also not generally OCA-proof
when~$\delta > \mu$: a collection of transaction creators all with the
same valuation~$v \in (r+\mu,r+\delta)$ would be better off in an OCA
with a miner than bidding on-chain according to an individually
rational bidding strategy (which would lead to the automatic exclusion
of all the transactions).  If $\delta-\mu$ is small, however, the
improvement in joint surplus possible through an OCA will be
negligible.
\end{remark}

Two of the biggest differences between an FPA and the 1559 mechanism
are the switch to a posted-price-type mechanism and the burning of
transaction fees.  The OCA-proofness considerations in our final two
corollaries explain why it's important to make both changes
simultaneously, rather than either one in isolation.

\begin{corollary}\label{cor:fpaburn_ocaproof}
For every~$\beta \in (0,1]$, a $\beta$-burn FPA is not OCA-proof.
\end{corollary}

\begin{proof}
Assume that $\mu =0$; the general case is similar.
Fix an (irrelevant) history $\history$, and
consider an arbitrary bidding strategy~$\sigma_{\history}$.
If~$\sigma_{\history}$ is the all-zero function, the allocation rule 
instructs the miner to include an arbitrary feasible subset of
transactions (breaking ties arbitrarily but consistently), which will
not generally be a joint utility-maximizing feasible subset.
So suppose that $\sigma_{\history}(\valt) > 0$ for some $\valt > 0$.
Then, if the mempool is a single transaction with
value~$\valt$, the joint utility of the on-chain outcome with
bid~$\sigma_{\history}(\valt)$ is $\valt - \beta \cdot
\sigma_{\history}(\valt)$.  
Because $\beta > 0$, 
the joint utility would be higher
with any smaller
on-chain bid $\bidt \in (0,\sigma_{\history}(\valt))$.
We conclude that no bidding strategy (individually rational
or otherwise) is guaranteed to maximize the joint utility, and hence
(by Proposition~\ref{prop:ocaproof}) this TFM is not OCA-proof.
\end{proof}

\begin{corollary}\label{cor:noburn1559_ocaproof}
For every~$\beta \in [0,1)$,
the $\beta$-burn 1559 mechanism is not OCA-proof.
\end{corollary}

\begin{proof}
Fix $\beta \in [0,1)$ and
a history $\history$ that leads to a positive base
fee~$r = \alpha(\history) > 0$.
Consider an arbitrary individually rational bidding
strategy~$\sigma_{\history}$.
Because transaction creators are charged their bids in
the $\beta$-burn 1559 mechanism, individual rationality implies that
there is no overbidding: $\sigma_{\history}(\valt) \le \valt$ for
every $\valt \ge 0$. 
Now consider a mempool with one transaction~$t$ with a
valuation~$\valt$ that is strictly between~$\beta r + \mu$ and $r +
\mu$, where~$\mu$ denotes the marginal cost; $\valt$ exists
because~$\beta < 1$ and $r > 0$.
Because~$\valt < r + \mu$,
the suggested bid~$\sigma_{\history}(\valt)$ would lead to~$t$'s
exclusion and a joint utility of~0.
The bid~$\bidt = r + \mu$ would lead to  a joint utility of
$\valt - \beta r - \mu$ which, because $\valt > \beta r + \mu$, is
strictly positive.  Proposition~\ref{prop:ocaproof} now implies that
the mechanism is not OCA-proof.
\end{proof}

\section{Discussion}\label{s:disc}

\paragraph{The 1559 mechanism vs.\ the tipless mechanism, and beyond.}
The 1559 and tipless mechanisms are the two clear winners in our study
of MMIC, DSIC, and OCA-proof transaction fee mechanisms---both provide
two of these three incentive-compatibility guarantees in all
circumstances, and all three in what is plausibly the common case of a
base fee that is reasonably well tuned to the current demand.  In a
block with an excessively low base fee (presumably due to rapidly
increasing demand), one might expect bidding wars in the 1559
mechanism (on account of the failure of DSIC) and off-chain coordination
in the tipless mechanism (due to the breakdown of OCA-proofness).  We
conjecture that there is no transaction fee mechanism that always
satisfies all three properties.\footnote{As discussed in
 Section~\ref{s:intro}, Chung and Shi~\cite{CS23} prove such an
  impossibility result under a collusion-resistance requirement different
  than OCA-proofness. The relative weakness
  of the OCA-proof condition appears to be the primary challenge
  in proving this conjecture.}
More generally, it would be
interesting to characterize the mechanisms that satisfy various
subsets and relaxations of these three properties.

Perhaps the strongest argument in favor of the tipless mechanism over
the 1559 mechanism is its simplicity.  On the user side, there are
several simplifications.  The creator of a transaction~$t$ must
specify only one parameter (a fee cap~$c_t$) rather than two (a fee
cap~$c_t$ and a tip~$\delta_t$).  The ``obvious optimal bid'' in the
tipless mechanism (setting~$c_t=\valt$) is optimal for every block and
no matter what the bids of the competing transactions.  The ``obvious
optimal bid'' in the 1559 mechanism (setting~$c_t=\valt$
and~$\delta_t=\mu$) is optimal only in blocks with neither an
excessively low base fee nor overbidding.  On the miner side, assuming
that~$\delta \ge \mu$, the revenue-maximizing strategy simplifies to
maximizing the block size while using only eligible transactions
(i.e., with bid at least~$r+\delta$).  On the negative side, the
hard-coded tip~$\delta$ in the tipless mechanism would likely need to
be adjusted over time via protocol upgrades.  Also, the tipless
mechanism's apparent DSIC advantage over the 1559 mechanism in blocks
with excessively low base fees breaks down in the presence of cartels
of colluding users or even a single user who creates multiple
transactions.  The reason is that, in a block with an excessively low
base fee, such a cartel or user could coordinate bids across
transactions to manipulate the (arbitrary but consistent) tie-breaking
rule of the mechanism.  When the base fee is not excessively low, the
1559 and tipless mechanisms are effectively unlimited-supply
posted-price mechanisms; as such, the bidding strategies in
Theorems~\ref{t:tipless_dsic} and~\ref{t:1559dsic} remain optimal for
users that control multiple transactions (assuming an additive
valuation over them) and for cartels of colluding users.

\vspace{-.1in}

\paragraph{Alternatives to burning.}
Corollary~\ref{cor:noburn1559_ocaproof} shows that, for a block's base
fee to be economically meaningful, no revenue from it can be passed on
to the block's miner.  The simplest solution is to burn all the base
fee revenues.  One appealing alternative implementation, with the same
incentive-compatibility properties, is to instead split the base fee
revenues of a block equally among the miners of the next~$\ell$
blocks.  (Here~$\ell$ is a tunable parameter; the~1559 mechanism can
be thought of as the special case in which~$\ell=0$.)  Thus, a miner
of a block receives a $1/\ell$ fraction of the sum of the base fee
revenues from the previous~$\ell$ blocks, along with all of the tips
from the current block.  While burning the base fee revenue favors
holders of the native currency (through deflation), this
``pay-it-forward'' implementation favors miners (through more
transaction fee revenue).  A second trade-off between the two
implementations concerns whether variability in demand (and hence
transaction fee revenue) translates to variability in blockchain
security or in the issuance of new currency.  With money-burning,
every block potentially changes the money supply in two ways: minting
new coins for inflationary rewards (like a block reward), and burning
the coins used to pay the base fee.  Thus the blockchain's inflation
rate would be variable and unpredictable from block to block, but
miner revenue (which effectively pays for the blockchain's
security~\cite{auer,budish}) would stay relatively constant (modulo
fluctuations in the market price of the native currency).  With the
pay-it-forward implementation, the inflation rate would be essentially
deterministic but the miner rewards (and, hence, blockchain security)
would be unpredictable (though never less than that with
money-burning).

\vspace{-.1in}

\paragraph{Multi-block time scales.}
This paper focuses on incentive issues in transaction fee mechanisms
at the time scale of a single block.  Many candidate deviations by
miners and users manifest already at this time scale.  The 1559
mechanism and its variants entangle the transaction fee mechanisms for
different blocks through a history-dependent base fee.  Dependencies
between blocks open up the possibility of miner deviations that unfold
over longer time scales.  

For example, publishing a
smaller-than-target block decreases the base fee for the next block,
potentially increasing the revenue of that block's miner.  Every miner
would happily free ride on previous miners who have sacrificed some
eligible transactions to keep block sizes and hence the base fee down,
but would rather not 
make such a sacrifice themselves.
Long-term manipulation by a cartel of miners thus 
resembles the challenge of
sustaining collusion in a repeated multi-player Prisoner's Dilemma
game: If all players cooperate (e.g., keep block sizes small to keep
the base fee low in the 1559 mechanism or the bids high in an FPA)
they are all better off, but each player has a myopic incentive to
unilaterally deviate from this strategy (e.g., when mining a block,
pack it full to maximize the immediate tip revenue, thereby
decreasing the revenue earned by miners of future blocks).  

Persistent
and harmful miner collusion has not yet been definitively observed in
a major blockchain such as Bitcoin or Ethereum.  None of the primary
transaction fee mechanisms discussed in this paper (FPAs, the 1559
mechanism, and the tipless mechanism) are obviously more vulnerable
than the others to such long-term collusion.  It would be interesting
to develop a more nuanced understanding of this issue, especially in a
regime in which a single miner or validator might control a significant
fraction (30\%, say) of the overall hashrate or stake and therefore might
plausibly benefit from non-myopic strategies.\footnote{Non-myopic
  validator strategies might be particularly effective in proof-of-stake
  blockchain protocols that, like the current version of Ethereum,
  grant validators some degree of advance knowledge about which blocks
  they will be chosen to produce.}

\end{document}